%% file: lbcs.tex
\documentclass[reqno,10pt]{amsart}		

\calclayout

\usepackage[english]{babel}
\usepackage{amsfonts, amsmath, amssymb, amsthm}
\usepackage{hyperref}

\newtheorem{lemma}{Lemma}
\newtheorem{proposition}{Proposition}
\theoremstyle{remark}
\newtheorem{remark}{Remark}

\usepackage{algorithm}
\usepackage[noend]{algpseudocode}

\usepackage{braket}
	\renewcommand\bra[1]{{\langle{#1}|}}
	\renewcommand\ket[1]{{|{#1}\rangle}}

\usepackage{multirow}
\usepackage{array}
    \newcolumntype{P}[1]{>{\centering\arraybackslash}p{#1}}
\newcommand{\rightvar}[1]{\multicolumn{1}{ r|}{ #1 }}

\usepackage{tikz}
\usepackage{pgfplots}
    \pgfplotsset{compat=1.14}

\def\C{\mathbb{C}}
\def\R{\mathbb{R}}
\DeclareMathOperator{\sgn}{sgn}
\DeclareMathOperator{\tr}{tr}
\DeclareMathOperator{\supp}{supp}
\DeclareMathOperator{\wt}{wt}
\def\E{\mathbb{E}} 
\DeclareMathOperator{\Var}{Var}

\title[]{Measurements of Quantum Hamiltonians with Locally-Biased Classical Shadows}

\usepackage[foot]{amsaddr}
\author{Charles Hadfield}
\email{charles.hadfield@ibm.com}
\author{Sergey Bravyi}
\author{Rudy Raymond}
\author{Antonio Mezzacapo}
\email{antonio.mezzacapo@ibm.com}

\address[SB, CH, AM]{IBM Quantum, IBM T.J. Watson Research Center, Yorktown Heights, NY 10598}
\address[RR]{IBM Quantum, IBM Research -- Tokyo, 19-21 Nihonbashi Chuo-ku, Tokyo 103-8510}
\address[RR]{Quantum Computing Center, Keio University, Hiyoshi 3-14-1, Kohoku, Yokohama 223-8522}

\begin{document}

\maketitle

\begin{abstract}
Obtaining precise estimates of quantum observables is a crucial step of variational quantum algorithms. We consider the problem of estimating expectation values of molecular Hamiltonians, obtained on states prepared on a quantum computer. We propose a novel estimator for this task, which is locally optimised with knowledge of the Hamiltonian and a classical approximation to the underlying quantum state. Our estimator is based on the concept of classical shadows of a quantum state, and has the important property of not adding to the circuit depth for the state preparation. We test its performance numerically for molecular Hamiltonians of increasing size, finding a sizable reduction in variance with respect to current measurement protocols that do not increase circuit depths.
\end{abstract}

\thispagestyle{empty}

\begin{section}{Introduction}

Estimating observables of interest for quantum states prepared on quantum processor is a central subroutine in a variety of quantum algorithms. Improving the precision of the measurement process is a pressing need, considering the fast-paced increase in size of current quantum devices. One key application is the energy estimation of complex molecular Hamiltonians, a staple of variational quantum eigensolvers (VQE)~\cite{peruzzo2014variational,o2016scalable,Kandala2017Hardware-efficientMagnets,hempel2018quantum}. 
Readout of quantum information on quantum processors is available only through single-qubit projective measurements. The outcomes of these single-qubit measurements are combined to estimate quantum observables described by linear combinations of Pauli operators. Naively, each Pauli operator can be estimated independently by appending a quantum circuit composed of one layer of single-qubit gates at the end of state preparation, before readout. 

A series of recent efforts \cite{jena2019pauli,yen2020measuring,huggins2019efficient,gokhale2019minimizing,zhao2019measurement,ryabinkin2020iterative,crawford2019efficient,hamamura2019efficient} has shown that savings in the number of measurements can be obtained for the estimation of complex observables, at the expense of increasing circuit depths. This increase in circuit depth can defy the purpose of variational quantum algorithms, which aim to keep gate counts low~\cite{preskill2018quantum}.

Other strategies, more amenable to execution on near-term devices, have considered reducing the number of measurements while simultaneously not increasing circuit depth. These strategies are based on grouping together Pauli operators that can be measured in the same single-qubit basis. This Pauli grouping approach was introduced in~\cite{Kandala2017Hardware-efficientMagnets} and explored thoroughly in~\cite{verteletskyi2020measurement} for chemistry systems. Machine learning techniques have also recently been used to tackle the measurement problem~\cite{torlai2019precise}, with no increase in circuit depth. The machine learning approach is based upon the assumption that fermionic neural-network states can capture quantum correlations in ground states of molecular systems~\cite{choo2020fermionic}. 

The measurement problem has been considered in the context of predicting collections of generic observables on reduced density matrices~\cite{Huang2020PredictingMeasurements, bonet2019nearly, cotler2020quantum}. The best asymptotic scalings up to poly-logarithmic factors are obtained in \cite{Huang2020PredictingMeasurements}, where it is proposed to characterise a quantum state through random measurements, later used to retrieve arbitrary observables.

In this article we introduce an estimator that recovers, in expectation, mean values of observables on quantum states prepared on quantum computers.
The protocol is based on \emph{classical shadows using random Pauli measurements} introduced in \cite{Huang2020PredictingMeasurements} and referred to as \emph{classical shadows} in this present article. We show how sampling from random measurement bases in the original protocol can be locally biased towards certain bases on each individual qubit. We name this technique \emph{locally-biased classical shadows}. We show how to optimise the estimator's local bias on each qubit based on the knowledge of a target observable and a classical approximation of the quantum state, named \emph{reference state}. We also prove that this optimisation has a convex cost function in certain regimes. We benchmark our optimisation procedure in the setting of quantum chemistry Hamiltonians, where reference states can be obtained from the Hartree-Fock solution, or multi-reference states, obtained with perturbation theory. We finally compare the variance of our estimator to previous methods for estimating average values that do not increase circuit depth, obtaining consistent improvements.

\begin{subsection}*{Outline of the paper}
Section~\ref{sec:shadows_uniform} reviews classical shadows using random Pauli measurements in a notation convenient to this current article. Section~\ref{sec:shadows_biased} provides the construction of the locally-biased classical shadows and calculates the expectation and variance associated with the estimator introduced. Section~\ref{sec:energy} shows how to optimise the estimator. Section~\ref{sec:experiments} benchmarks our estimator for molecular energies on molecules of increasing sizes. Section~\ref{sec:conclusion} finishes with closing remarks. Appendix~\ref{app:comparative_algorithms} reviews the methods for molecular energy estimation to which we compare our estimator.
\end{subsection}

\begin{subsection}*{Acknowledgements}
We thank Giacomo Nannicini for useful discussions regarding the convexity of the cost functions introduced here.
SB acknowledges the support of the IBM Research Frontiers Institute.
\end{subsection}

\end{section}

\begin{section}{Classical Shadows Using Random Pauli Measurements}\label{sec:shadows_uniform}

Classical shadows using random Pauli measurements has been introduced in \cite{Huang2020PredictingMeasurements}. This section reproduces the procedure in a different style. Since we are only concerned with estimating one specific observable, we do not mention the aspect of a snapshot, nor the efficient description using the symplectic representation, nor the notion of median of means.

The problem that we want to address is the estimation of $\tr(\rho O)$ for a given $n$-qubit state $\rho$ and an observable $O$ decomposed as a linear combination of Pauli terms:
\begin{equation}
    O = \sum_{Q \in \{I,X,Y,Z\}^{\otimes n}}
            \alpha_Q Q
\end{equation}
where $\alpha_Q\in\R$. 
Notationally, for a Pauli operator $Q$ as above and for a given qubit $i\in\{1,2,\dots,n\}$ we shall write $Q_i$ for the $i$\textsuperscript{th} single-qubit Pauli operator so that $Q=\otimes_i Q_i$. We denote the support of such an operator
$\supp(Q) = \{i | Q_i \neq I\}$ and its weight $\wt(Q)= |\supp(Q)|$. An $n$-qubit Pauli operator
$Q$ is said to be full-weight if $\wt(Q)=n$.


This task of estimating $\tr(\rho O)$ is accomplished with classical shadows of~\cite{Huang2020PredictingMeasurements} as described in Algorithm~\ref{alg:pauli_shadows_uniform}. 
Briefly, one randomly selects a Pauli basis for each of the $n$-qubits in which to measure the quantum state; this is irrespective of the operator $O$. Then, after measurement, non-zero estimates can be provided for all Pauli operators which qubit-wise commute with the measurement bases. All other Pauli operators are implicitly provided with the zero estimator for their expectation values.

We introduce the function from \cite[Eq. E28]{Huang2020PredictingMeasurements}. For two $n$-qubit Pauli operators $P,Q$ define, for each qubit $i$,
\begin{equation}\label{eq:f_from_HKP}
    f_i(P,Q)
    =   
    \begin{cases}
        1 & \textrm{if $P_i=I$ or $Q_i=I$}; \\
        3 & \textrm{if $P_i=Q_i \neq I$}; \\
        0 & \textrm{else}.
    \end{cases}
\end{equation}
and extend this to the multi-qubit setting by declaring $f(P,Q) = \prod_{i=1}^n f_i(P,Q)$. Also, given a full-weight Pauli operator $P$, we let $\mu(P,i)\in\{\pm 1\}$ denote the eigenvalue measurement when qubit $i$ is measured in the $P_i$ basis. For a subset $A\subseteq\{1,2,\dots, n\}$ declare
\begin{equation}
	\mu(P,A) = \prod_{i\in A} \mu(P,i)
\end{equation}
with the convention that $\mu(P,\varnothing)=1$.

\begin{algorithm}
	\caption{Estimation of observable via (uniform) classical shadows}
	\label{alg:pauli_shadows_uniform}
	\begin{algorithmic}
		\For{sample $s \in \{1, 2, \dots, S\}$}
			\State Prepare $\rho$;
			\State Uniformly at random pick $P \in \{X,Y,Z\}^{\otimes n}$;
			\For{qubit $i\in\{1,2,\dots, n\}$}
				\State Measure qubit $i$ in $P_i$ basis providing eigenvalue measurement $\mu(P,i) \in \{\pm 1\}$;
			\EndFor
		\State Estimate observable expectation
			\[
			\nu^{(s)} = \sum_Q \alpha_Q f(P,Q) \mu(P,\supp(Q))
			\]
		\EndFor
		\Return $\nu = \frac 1 S \sum_s \nu^{(s)}$.
	\end{algorithmic}
\end{algorithm}
As shown in~\cite{Huang2020PredictingMeasurements}, the output of this algorithm is an unbiased estimator of the desired expectation value, that is, $\E(\nu)=\tr{(\rho O)}$.

\end{section}

\begin{section}{Locally-Biased Classical Shadows}\label{sec:shadows_biased}

In this section we generalise classical shadows by observing that the randomisation procedure of the Pauli measurements can be biased in the measurement basis for each qubit. We build an estimator based on biased measurements, which in expectation recovers $\tr(\rho O)$. We then proceed to calculate its variance.

As in Section~\ref{sec:shadows_uniform}, we wish to estimate $\tr(\rho O)$ for a given state $\rho$ and an observable  $O=\sum_Q \alpha_Q Q$. For each qubit $i \in \{1,2,\dots,n\}$, consider a probability distribution $\beta_i$ over $\{X,Y,Z\}$ and denote by $\beta_{i}(P_i)$ the probability associated with each Pauli $P_i\in\{X,Y,Z\}$. We write $\beta$ for the collection $\{ \beta_i \}_{i=1}^n$ and note that $\beta$ may be considered a probability distribution on full-weight Pauli operators by associating with $P\in \{X,Y,Z\}^{\otimes n}$ the probability $\beta(P) = \prod_i \beta_{i}(P_i)$.

We generalise the function introduced in Eq.~\eqref{eq:f_from_HKP}. For two $n$-qubit Pauli operators $P,Q$ and a product probability distribution $\beta$, define, for each qubit $i$,
\begin{equation}
    f_i(P,Q,\beta)
    =   
    \begin{cases}
        1 & \textrm{if $P_i=I$ or $Q_i=I$}; \\
        (\beta_i(P_i))^{-1} & \textrm{if $P_i=Q_i \neq I$}; \\
        0 & \textrm{else}.
    \end{cases}
\end{equation}
In the above, $(\beta_i(P_i))^{-1}$ ought be interpreted as $0$ if $\beta_i(P_i)$ vanishes. We extend this to the multi-qubit setting by declaring
\begin{equation}
    f(P,Q, \beta) 
    = 
    \prod_{i=1}^n f_i(P,Q,\beta).
\end{equation}

Algorithm~\ref{alg:pauli_shadows_biased} describes an estimator via locally-biased classical shadows. Note that the (uniform) classical shadows case is retrieved when $\beta_i(P_i) = \frac13$ for every qubit $i\in\{1,2,\dots, n\}$ and every Pauli term $P_i\in\{X,Y,Z\}$.

\begin{algorithm}
	\caption{Estimation of observable via locally-biased classical shadows}
	\label{alg:pauli_shadows_biased}
	\begin{algorithmic}
		\For{sample $s \in \{1,2, \dots, S\}$}
			\State Prepare $\rho$;
			\For{qubit $i\in \{1,2,\dots, n\}$}
				\State Randomly pick $P_i\in\{X,Y,Z\}$ from $\beta_i$-distribution;
				\State Measure qubit $i$ in $P_i$ basis providing eigenvalue measurement $\mu(P,i) \in \{\pm 1\}$;
			\EndFor
			\State Set $P=\otimes_{i=1}^n P_i$;
			\State Estimate observable expectation
			\[
			\nu^{(s)} = \sum_Q \alpha_Q f(P,Q, \beta) \mu(P,\supp(Q))
			\]
		\EndFor
		\Return $\nu = \frac 1 S \sum_s \nu^{(s)}$.
	\end{algorithmic}
\end{algorithm}
Algorithm~\ref{alg:pauli_shadows_biased} recovers the expectation $\tr(\rho O)$, as shown in the following lemma. 

\begin{lemma}\label{lem:estimator_calculations}
	The estimator $\nu$ from Algorithm~\ref{alg:pauli_shadows_biased} with a single sample $(S=1)$ satisfies
	\begin{equation}\label{eq:moments}
		\E(\nu) = \sum_Q \alpha_Q \tr(\rho Q)
		\qquad
		\textrm{and}
		\qquad
		\E(\nu^2) = \sum_{Q,R} f(Q,R,\beta) \alpha_Q\alpha_R \tr(\rho QR).
	\end{equation}
\end{lemma}

\begin{proof}
Let $\E_P$ denote the expected value over the distribution $\beta(P)$.
Let $\E_{\mu(P)}$ denote the expected value over the measurement outcomes
for a fixed Pauli basis $P$.
Using the fact that $\beta(P)$ is a product distribution one can easily check that 
\begin{equation}\label{eq:meanP1}
	\E_P f(P,Q,\beta) = 1,
\end{equation}
\begin{equation}\label{eq:meanP2}
	\E_P f(P,Q,\beta) f(P,R,\beta)  = f(Q,R,\beta)
\end{equation}
for any $Q,R\in \{I,X,Y,Z\}^{\otimes n}$.

Let us say that an $n$-qubit Pauli operator $Q$ agrees with a basis $P\in \{X,Y,Z\}^{\otimes n}$
iff $Q_i\in \{I,P_i\}$ for any qubit $i$.
Note that $f(P,Q,\beta)=0$ unless $Q$ agrees with $P$.
For any $n$-qubit Pauli operators $Q,R$ that agree with a basis $P$ one has
\begin{equation}\label{eq:meanM1}
	\E_{\mu(P)} \mu(P,\mathrm{supp}(Q)) =\tr{(\rho Q)}
\end{equation}
and
\begin{equation}\label{eq:meanM2}
	\E_{\mu(P)} \mu(P,\mathrm{supp}(Q)) \mu(P,\mathrm{supp}(R))=\tr{(\rho QR)}.
\end{equation}
To get the last equality, observe that $\mu(P,A)\mu(P,A')=\mu(P,A\oplus A')$
for any subsets of qubits $A,A'$, where $A\oplus A'$ is the symmetric difference of
$A$ and $A'$. The assumption that both $Q$ and $R$ agree with the same basis $P$
implies that $\mathrm{supp}(Q)\oplus \mathrm{supp}(R)=\mathrm{supp}(QR)$.
Now Eq.~\eqref{eq:meanM2} follows from Eq.~\eqref{eq:meanM1}.

By definition, the expected value in Eq.~\eqref{eq:moments} is a composition
of the expected values over a Pauli basis $P$ and over the measurement outcomes $\mu(P)$, that is,
$\E = \E_P \E_{\mu(P)}$.
Using the above identities one gets  
\[
\E(\nu)=  \E_P \E_{\mu(P)} \nu = \sum_Q \alpha_Q \tr{(\rho Q)} \E_P f(P,Q,\beta)
= \sum_Q \alpha_Q \tr{(\rho Q)}
\]
Here the second equality is obtained using Eq.~\eqref{eq:meanM1}
and the linearity of expected values. 
The third equality follows from Eq.~\eqref{eq:meanP1}.
Likewise,
\[
\E(\nu^2) =  \E_P \E_{\mu(P)} \nu^2 = \sum_{Q,R} \alpha_Q \alpha_R \tr{(\rho QR)} \E_P f(P,Q,\beta) f(P,R,\beta)
=  \sum_{Q,R} \alpha_Q \alpha_R  f(Q,R,\beta) \tr{(\rho QR)}.
\]
Here the second equality is obtained using Eq.~\eqref{eq:meanM2} and observing that
that $f(P,Q,\beta)f(P,R,\beta)=0$ unless both $Q$ and $R$ agree with $P$.
The third equality follows from Eq.~\eqref{eq:meanP2}.
\end{proof}

Recall that in the context of using a quantum processor, we aim to use the random variable $\nu$ to estimate $\tr(\rho O)$ to some (additive) precision $\varepsilon$. This dictates the number of samples $S$ required. Specifically, for fixed $\rho,O$, we require $S=O(\varepsilon^{-2} \Var(\nu^{(s)}))$ where $\Var(\nu^{(s)})$ is obviously independent of the specific sample $s$. For future reference, we record explicitly the variance of $\nu$ for a single sample $(S=1)$. Lemma~\ref{lem:estimator_calculations} establishes
\begin{equation}\label{eq:var_pauli_biased}
	\Var(\nu) = \left(\sum_{Q,R} f(Q,R,\beta) \alpha_Q \alpha_R \tr(\rho Q R)\right) - \left(\tr(\rho O) \right)^2.
\end{equation}
\begin{remark}
In \cite[Proposition 3]{Huang2020PredictingMeasurements}, the authors aim to upper-bound this variance independently of the state $\rho$. In the uniform setting, this is achieved with an application of Cauchy-Schwarz and it leads to a bound of $4^k \|O\|^2_\infty$ where $k$ is the weight of the operator.
\end{remark}

\end{section}

\begin{section}{Optimised Locally-Biased Classical Shadows}\label{sec:energy}

In this section we show how the locally-biased classical shadows introduced in Section~\ref{sec:shadows_biased} can be optimised when one has partial knowledge about the underlying quantum state. This partial information is obtained efficiently with a classical computation. This is the case of VQE for molecular Hamiltonians if one initialises the variational procedure from a reference state, which can be, for example, the Hartree-Fock solution, a generic fermionic Gaussian state~\cite{dallaire2019low}, or perturbative M{\o}ller-Plesset solutions. Our method can be also extended to generic many-body Hamiltonians, considering for example the generation of reference states with semidefinite programming~\cite{bravyi2019approximation}. On a more general note, the existence of a good reference state is the assumption of all algorithms that target ground state properties of interacting many-body problems, including quantum phase estimation. 

In this setting, we optimise the probability distributions $\beta = \{\beta_i\}_{i=1}^n$ to obtain the smallest variance on a given reference state. To do this, we consider the variance calculated in Eq.~\eqref{eq:var_pauli_biased} and extract from it the component which, associated with the reference state, explicitly depends on the distributions $\beta$. We proceed to optimise this cost function, thereby minimising the variance, noting that a negligible restriction of the cost function that we use leads to a convex optimisation problem. Finally, we use the optimised distributions $\beta^*$ to build molecular energy estimators as defined in Algorithm~\ref{alg:pauli_shadows_biased}. 

To set notation, we introduce a molecular Hamiltonian, $H$, acting on $n$ qubits. We write
\begin{align}
    H = \sum_{P \in \{I,X,Y,Z\}^{\otimes n}}
            \alpha_P P
\end{align}
and denote by $H_0$ the traceless part of $H$.

\begin{subsection}{Single-reference optimisation}

We first consider the case in which the reference state is a single logical basis state. This is the case if a VQE targeting a molecular Hamiltonian in the molecular basis is initialised with the Hartree-Fock state. Motivated by this, we use the label ``HF" to indicate the single logical basis state. However we remark that the results here can be generalised out of the quantum chemistry domain.
We are given a reference product state $\rho_{\mathrm{HF}} = \frac1{2^n}\otimes_{i=1}^n (I+m_i Z)$ where $m_i\in\{\pm1\}$.
The variance of the estimator $\nu$ is independent of the constant term $H-H_0$. Writing the variance from Eq.~\eqref{eq:var_pauli_biased} for the state $\rho_{\mathrm{HF}}$ upon explicit removal of the constant term reads
\begin{equation}\label{eq:var_given_hf}
    \Var(\nu | \rho_\mathrm{HF})
    =
    \sum_{Q \neq I^{\otimes n}}
    \sum_{R \neq I^{\otimes n}}
        f(Q,R,\beta)
        \alpha_Q \alpha_R
        \tr(\rho_{\mathrm{HF}}QR)
    - \tr(\rho_{\mathrm{HF}} H_0 )^2.
\end{equation}
Our objective is to find probability distributions $\beta$ so as to minimise Eq.~\eqref{eq:var_given_hf}. The following proposition explicits the relevant cost function appropriate to this task.

\begin{proposition}
Given a reference product state $\rho_{\mathrm{HF}}$, represented by the logical basis element $\{m_i\}_{i=1}^n$, the variance associated with Algorithm~\ref{alg:pauli_shadows_biased} is minimised upon choosing $\beta$ so as to minimise
\begin{align}\label{eq:cost}
    \mathrm{cost}(\beta)
    =
    \sum_{(Q,R)\in\mathcal{I}_{Z^{\otimes n}}}
        \alpha_Q
        \alpha_R
        \prod_{i | Q_i=R_i\neq I}
            (\beta_i(Q_i))^{-1}
        \prod_{i | Q_i\neq R_i}
            m_i
\end{align}
subject to $\beta_{i,P} \ge 0$ and $\beta_{i,X}+\beta_{i,Y}+\beta_{i,Z}=1$ for all $i$. In the above, the sum is taken over ``influential pairs":
\begin{align}
    \mathcal{I}_{Z^{\otimes n}}
    =
    \left\{\left.
        (Q,R)
        \,\right|\,
        \textrm{$Q,R\neq I^{\otimes n}$ and for all $i$, either $Q_i=R_i$, or $\{Q_i,R_i\}=\{I,Z\}$}
    \right\}
\end{align}
\end{proposition}
\begin{proof}
We must pay attention to only $\beta$-dependent terms in $\Var(\nu | \rho_\mathrm{HF})$.
The simple structure of $\rho_{\mathrm{HF}}$ implies, for $n$-qubit non-identity Pauli operators $Q,R$,
\begin{align*}
    f(Q, R, \beta)\tr(\rho_{\mathrm{HF}} Q R)
    &=
    \prod_{i=1}^n 
        f_i(Q, R, \beta) 
        \tr\left(
                \frac12(I+m_iZ)Q_iR_i
            \right)
    \\
    &=
    \prod_{i=1}^n
    f_i(Q, R, \beta)\delta_{Q_i,R_i}
    +
    m_i\left(
        \delta_{Q_i,Z}\delta_{R_i,I}
        +
        \delta_{Q_i,I}\delta_{R_i,Z}
    \right)
    \\
    &=
    \prod_{i=1}^n
    \delta_{Q_i,R_i}
    \left(
        \delta_{Q_i,I} + (1-\delta_{Q_i,I})(\beta_i(Q_i))^{-1}
    \right)
    +
    m_i\left(
        \delta_{Q_i,Z}\delta_{R_i,I}
        +
        \delta_{Q_i,I}\delta_{R_i,Z}
    \right)
\end{align*}
The preceding display is independent of $\beta$ whenever $(Q,R)\not\in \mathcal{I}_{Z^{\otimes n}}$. Hence the cost function captures precisely the component of the variance (when estimating the reference product state) which is dependent on the probability distributions $\beta$.
\end{proof}

Some remarks are in order.
\begin{remark}
The cost function of Eq.~\eqref{eq:cost} is \emph{not} convex. If we however restrict to diagonal terms from the set of influential pairs, then we obtain the following alternative cost function, which we refer to as the \emph{diagonal} cost function:
\begin{align}\label{eq:cost_diag}
    \mathrm{cost}_\mathrm{diag}(\beta)
    =
    \sum_Q \alpha_Q^2 \prod_{i\in\supp(Q)} (\beta_i(Q_i))^{-1}
\end{align}
This diagonal cost function is convex: For fixed $Q$, the function $-\log \beta_i(Q_i)$ is convex, hence so too is $\sum_{i\in\supp(Q)} (-\log(\beta_i(Q_i)))$. Exponentiating this result implies $\prod_{i\in\supp(Q)} \beta_i(Q_i)^{-1}$ is convex. The positive linear combination over Pauli operators $Q$ preserves convexity. In this convex case we are assured that the minimised collection of distributions provides a global minimum (of the diagonal cost function).

This diagonal cost function makes no reference to the specific single-reference state, and therefore can be used to find optimal $\beta_i$ which are independent of the underlying quantum state $\rho_\textrm{HF}$.
In fact, the diagonal cost function can be derived from Eq.~\eqref{eq:var_pauli_biased} when $\rho$ is the maximally mixed state. 
We also find that the diagonal cost function does however numerically give very satisfying results. We explain this in the following paragraph by relating the two cost functions.

Set 
$\Gamma(P,Q)
=
|\alpha_P|\,|\alpha_Q|\,
\prod_{i|P_i = Q_i \neq I} 
(\beta_i(P_i))^{-1} 
$. 
The diagonal cost function is $\sum_{Q} \Gamma(Q,Q)$, while the original cost function is at most $\sum_{P,Q} \Gamma(P,Q)$, (the sum is over only influential pairs in the original cost function). By definition, the following inequality holds
\begin{align}
    \Gamma(P,Q) 
    \le 
    \frac{\Gamma(P, P) + \Gamma(Q, Q)}{2}.
\end{align}
Summation over all pairs $(P,Q)$ on both sides leads to
\begin{align}
    \sum_{P,Q} \Gamma(P,Q) 
    \le 
    |H_0|
    \sum_{Q} \Gamma(Q,Q)
\end{align}
where $|H_0|$ is the number of traceless terms in the Hamiltonian. The preceding display upper-bounds the original cost function, therefore minimising the diagonal cost function implies minimising the original cost function per Pauli in the Hamiltonian.

\end{remark}

\begin{remark}
The diagonal cost function can be formulated in the language of geometric programming \cite{Boyd2007}, while the original cost function is an example of signomial geometric programming. 
\end{remark}

\begin{remark} We solve these optimisation problems using the method of Lagrange multipliers. Specifically given current values $\beta^{(t)}(P)$ and update step-size $\Delta\in(0,1)$, we may update iteratively:
\begin{align}
    \beta^{(t+1)}(P) = (1-\Delta) \beta^{(t)}(P) + \Delta \beta^\mathrm{closed}(P)
\end{align}
where the closed-form Lagrange equations (detailed below) are calculated using values of $\beta^{(t)}(P)$ and must hold at optimality.
The closed-form equations for the diagonal cost function of Eq.~\eqref{eq:cost_diag} are
\begin{align}
    \beta_{i}(P_i) 
    = 
    \frac{  \sum_{Q | Q_i = P_i} 
                \alpha_Q^2 
                \prod_{j \in \supp(Q)} 
                    \beta_j(Q_j)^{-1} 
    }{      \sum_{Q | Q_i \neq I} 
                \alpha_Q^2 
                \prod_{j \in \supp(Q)} 
                    \beta_j(Q_j)^{-1}
    }
\end{align}
while for the original cost function of Eq.~\eqref{eq:cost}, they read
\begin{align}
    \beta_i(P_i) 
    = 
    \frac{  \sum_{   (Q,R) 
                    \in 
                    \mathcal{I}_{Z^{\otimes n}}
                    |
                    Q_i = R_i = P_i}
                \alpha_Q \alpha_R
                \prod_{j|Q_j = R_j \neq I} 
                    \beta_j(Q_j)^{-1} 
                \prod_{j|Q_j \neq R_j} 
                    m_j
    }{      \sum_{   (Q,R) 
                    \in 
                    \mathcal{I}_{Z^{\otimes n}}
                    |
                    Q_i = R_i \neq I} 
                \alpha_Q \alpha_R
                \prod_{j|Q_j = R_j \neq I} 
                    \beta_j(Q_j)^{-1} 
                \prod_{j|Q_j \neq R_j} 
                    m_j
    }
\end{align}

The iterative updates find optimal probability distributions for the diagonal cost function because at every iteration the constraints on $\beta$ are always satisfied whenever initialisation occurs with a random collection of probability distributions.
\end{remark}

\end{subsection}

\begin{subsection}{Multi-reference optimisation}

We finish this section by observing that the technique of optimising the probability distributions also works for multi-reference frame states such as fermionic Gaussian states, or perturbative solutions. Specifically, consider a multi-reference state $\ket{\psi}_\textrm{MR}$, written in the logical basis
\begin{align}
    \ket{\psi}_\textrm{MR} = \sum_{k=1}^K \lambda_k  \ket{\psi^{(k)}},
    \qquad
    \ket{\psi^{(k)}} = \ket{ b_1^{(k)} \cdots b_n^{(k)} }
\end{align}
where $b_i^{(k)}\in\{0,1\}$ are associated with $Z$-eigenvalues $m_i^{(k)}=(-1)^{b_i^{(k)}}$ and $\lambda_k\in\C$ are amplitudes such that $\ket{\psi}_\textrm{MR}$ is normalised. The associated density now reads
\begin{align}
    \rho_\textrm{MR}
    =
    \sum_{k,\ell}
        \lambda_k \overline{\lambda_\ell}
        \rho^{(k,\ell)},
    \qquad
    \rho^{(k,\ell)} 
    = 
    \otimes_{i=1}^n \ket{b_i^{(k)}} \bra{b_i^{(\ell)}}
\end{align}
In the following paragraphs, we calculate an appropriate cost function for this case.

Let us restrict ourselves to the single-qubit setting briefly:
$\rho^{(k,\ell)} = \ket{b^{(k)}} \bra{b^{(\ell)}}$.
There are two cases for $\rho^{(k,\ell)}$ dependent on whether $b^{(k)}, b^{(\ell)}$ agree or not. If they agree then 
$\rho^{(k,\ell)} = \frac12(I + (-1)^{b^{(k)}}Z)$.
If they disagree, then 
$\rho^{(k,\ell)} = \frac12(X + (-1)^{b^{(k)}}iY)$.
In a similar way to the single-reference setting, we need to calculate $f(Q,R,\beta) \tr(\rho^{(k,\ell)}QR)$. This is best done by considering the two cases:
We introduce the function $g$ when $b^{(k)}= b^{(\ell)}$ and obtain
\begin{align*}
    f(Q,R,\beta) \tr(\rho^{(k,\ell)}QR)
    &=
        \delta_{Q,R}
        \left(
            \delta_{Q,I} + (1-\delta_{Q,I})\beta_{Q}^{-1}
        \right)
        +
        (-1)^{b^{(k)}}
        \left(
        \delta_{Q,Z}\delta_{R,I}
        +
        \delta_{Q,I}\delta_{R,Z}
        \right)
    \\
    &= g(Q,R,\beta, b^{(k)});
\end{align*}
We introduce the function $h$ when $b^{(k)}\neq b^{(\ell)}$ and obtain
\begin{align*}
    f(Q,R,\beta) \tr(\rho^{(k,\ell)}QR)
    &=
        \left(
            \delta_{Q,X}\delta_{R,I} + \delta_{Q,I} \delta_{R,X}
        \right)
        +(-1)^{b^{(k)}} i
        \left(
            \delta_{Q,Y}\delta_{R,I} + \delta_{Q,I} \delta_{R,Y}
        \right)
    \\
    &= h(Q,R, b^{(k)}).
\end{align*}

We can now return to the multi-qubit setting to write down a cost function which ought be minimised:
\begin{align}
    \mathrm{cost}_\textrm{multi-ref}(\beta)
    =
    \sum_{k,\ell} 
        \lambda_k \overline{\lambda_\ell}
        \sum_{Q,R}
            \left(
            \prod_{i | b_i^{(k)} = b_i^{(\ell)}}
                g(Q_i,R_i,\beta_i, b_i^{(k)})
            \prod_{i | b_i^{(k)} \neq b_i^{(\ell)}}
                h(Q_i,R_i, b_i^{(k)})
            \right)
\end{align}

\end{subsection}

\end{section}

\begin{section}{Numerical experiments on molecular Hamiltonians}\label{sec:experiments}

In this section we test numerically the locally-biased classical shadows (LBCS) estimator defined in Algorithm~\ref{alg:pauli_shadows_biased} for molecular Hamiltonians. We consider six Hamiltonians corresponding to different molecules, represented in a minimal STO-3G basis, ranging from 4 to 16 spin orbitals. (The 8 qubit H$_2$ example uses a 6-31G basis.) We map the molecular Hamiltonians to qubit ones, using three encodings detailed in~\cite{bravyi2017tapering}. The result is qubit Hamiltonians defined on up to 16 qubits. The molecular Hamiltonians are defined in the molecular basis. In this basis, the Hartree-Fock state is a computational basis state. We choose the Hartree-Fock state as our single-reference state, and optimise the distributions $\beta$ according to Eq.~\eqref{eq:cost} and Eq.~\eqref{eq:cost_diag} separately. We call the optimisation procedure of the $\beta$ according to Eq. \eqref{eq:cost_diag} \emph{diagonal}. We then use the optimised $\beta^*$ to compute the  variance Eq.~\eqref{eq:var_pauli_biased} on the ground state of the molecular Hamiltonians; the ground state and the ground energy are obtained by the Lanczos method for sparse matrices. We report the results in Table~\ref{tab:algorithms}. 
In this table, we compare variances obtained with our LBCS estimator against other previously known observable estimators that do not increase circuit depth:
\begin{itemize}
    \item An estimator based on $\ell^1$ sampling of the Hamiltonian, detailed in~\cite{Wecker2015ProgressAlgorithms,arrasmith2020operator}. 
    \item An estimator which measures together collections of qubit-wise commuting Pauli operators. To find the collections of Pauli operators, we use a largest degree first (LDF) heuristic \cite{WelshPowell67}. The collections are then sampled according to their Hamiltonian $\ell^1$ weights.
    \item  Classical shadows as given in~\cite{Huang2020PredictingMeasurements}, which corresponds to the case $\beta_i(P_i)=\frac13$ for any qubit $i$ and Pauli term $P_i\in\{X,Y,Z\}$.
\end{itemize}
Details of the first two estimators may be found in Appendix~\ref{app:comparative_algorithms}.\footnote{Code is available upon request.}
For all the estimators, we report variances exactly computed on the ground states of the Hamiltonians considered.

\input{table_variances}

In all but one experiment of Table~\ref{tab:algorithms}, we observe that the LBCS estimator outperforms the other estimators. The one case where the LDF decomposition provides a lower variance -- H$_2$ on a minimal basis -- should be considered a curiosity due to the small qubit count.

We also plot in Figure~\ref{fig:distributions} an optimised distribution $\beta^*$. Specifically, we take the example of H$_2$O on 14 qubits in the Jordan-Wigner encoding. Due to the symmetry \cite{bravyi2017tapering} where the first 7 qubits correspond to spin-up orbitals, and the last 7 qubits correspond to spin-down orbitals, we observe that $\beta_{i}^* = \beta_{i+7}^*$ for $i\in\{1,2,\dots, 7\}$. Note also that the probabilities are symmetric in $X$ and $Y$ (which is not the case for the Bravyi-Kitaev encoding).

\input{figure_distributions}

We next analyse the role played by the specific fermionic encoding used. For a restricted set of Hamiltonians, Table~\ref{tab:encodings} reports variances for the three estimators: LDF grouping; classical shadows; and LBCS, with three different fermion-to-qubit encodings: Jordan-Wigner; parity; and Bravyi-Kitaev. Note that the variances for parity and Bravyi-Kitaev mappings are higher because those mappings generate Pauli distributions that tend to have more $X$ and $Y$ operators, as opposed to the linear tail of $Z$ operators of the Jordan-Wigner, against which the distributions $\beta$ can be easily biased. We do not report $\ell^1$ sampling in Table~\ref{tab:encodings}, as it is invariant under choice of encoding. Irrespective of the encoding, the locally biased classical shadows shows a reduction in variance over the LDF grouping whose collections are sampled according to their 1-norm.

\input{table_encodings}

The two different cost functions used to optimise the $\beta$-distributions provide very similar variance. This is remarkable considering that the diagonal cost function defined in Eq.~\eqref{eq:cost_diag} is convex. For any given molecule, our numerical analysis indicated that the non-convex cost function Eq.~\eqref{eq:cost} always converged to a single collection of distributions, irrespective of the initialised values for the distributions.

\end{section}

\begin{section}{Conclusion}\label{sec:conclusion}

This article has considered the measurement problem associated with molecular energy estimation on quantum computers and has proposed a new algorithm for that problem. Investigating the principal subroutine present in \emph{classical shadows using random Pauli measurements}, we are able to produce a non-uniform version of these shadows, termed \emph{locally-biased classical shadows}. These locally biased classical shadows require probability distributions for each qubit. By solving a convex optimisation problem for a given molecular Hamiltonian, we find appropriate probability distributions for measuring states which are close to the true ground state of the molecular Hamiltonian. We benchmark the proposed algorithm on systems up to 16 qubits in size and observe significant and consistent improvement over Pauli grouping heuristic algorithms. To claim this, we have compared with the LDF heuristic, noting that Ref.~\cite{verteletskyi2020measurement} finds that other heuristics produce a number of qubit-wise commuting sets that only differ by 10\%. We are able to obtain this improvement without solving computationally-intensive problems. This is unlike Pauli grouping methods which use node colouring and minimum clique covering, whose running times are quadratic in the number of Pauli terms of the Hamiltonian.

Finally, the introduction of such a domain-specific cost function is, to the authors' knowledge, novel. It is sufficiently general that applications of this idea will also be relevant in fields unrelated to quantum chemistry.

\end{section}



\bibliographystyle{unsrt}
\bibliography{references}

\appendix
\begin{section}{Comparative Algorithms for Estimating Molecular Hamiltonians}\label{app:comparative_algorithms}

This appendix provides details for the two algorithms against which we benchmark locally-biased classical shadows. Recall that we assume that the molecular Hamiltonian $H$ acts on $n$-qubits and that a given state $\rho$ is provided and whose energy we aim to estimate. We write
\begin{align}
    H = \sum_{P \in \{I,X,Y,Z\}^{\otimes n}}
            \alpha_P P
\end{align}
and denote by $H_0$ the traceless part of $H$. Denote by $\|\alpha\|_{\ell^1}$ the $\ell^1$-norm of the traceless coefficients, and associate with this norm the following $\ell^1$-distribution $\gamma$ over the Pauli operators:
\begin{align}
    \|\alpha\|_{\ell^1}
        &= \sum_{P\in\{I,X,Y,Z\}^{\otimes n}\backslash\{I^{\otimes n}\}} | \alpha_P |
    &
    \gamma(P)
        &= \frac 1 {\| \alpha\|_{\ell^1}}
            |\alpha_P|
\end{align}

We expose the dependence of the algorithms on the identity coefficient $\alpha_{I^{\otimes n}}$. This is because in the practical setting of molecular Hamiltonians considered in this text, the identity coefficient can be on the order of 10\% of $\|\alpha\|_{\ell^1}$. For the $\ell^1$ sampling, it would be unwise to prepare $\rho$ only to subsequently measure no qubits. For the largest degree first setting, it would be unwise to arbitrarily associate the identity operator to one of the collections of qubit-wise commuting Pauli operators thereby associating the identity operator's weight $|\alpha_{I^{\otimes n}}|$ to the corresponding collection's weight and overly favouring the sampling of said collection.

Recall the notation from Section~\ref{sec:shadows_uniform}. Given a Pauli operator $P$, we let $\mu(P,i)\in\{\pm 1\}$ denote the eigenvalue measurement when qubit $i$ is measured in the $P_i$ basis. For a subset $A\subseteq\{1,2,\dots, n\}$ we write $\mu(P,A) = \prod_{i\in A} \mu(P,i)$.

\begin{subsection}{Ell-1 algorithm}

This algorithm was the first algorithm proposed for estimating energies in the context of variational quantum algorithms \cite{Wecker2015ProgressAlgorithms}.
The $\ell^1$-norm of the traceless coefficients provides the probability distribution $\gamma$. We may use this probability distribution to select a Pauli operator $P$ which dictates the Pauli basis in which to measure the state $\rho$, thereby providing an estimate for $\tr(\rho P)$. Algorithm~\ref{alg:ell_1} describes this procedure precisely.

\begin{algorithm}
	\caption{Energy estimation via $\ell^1$-distribution over Pauli bases}
    \label{alg:ell_1}
	\begin{algorithmic}
	    \For{sample $s \in \{1, \dots, S\}$}
	        \State Prepare $\rho$;
	        \State Randomly pick $P$ from $\gamma$-distribution;
	        \For{qubit $i\in \supp(P)$}
	            \State Measure qubit $i$ in $P_i$ basis providing eigenvalue measurement $\mu(P,i) \in \{\pm1\}$;
	        \EndFor
	        \State Estimate observable expectation
	            \[
	            \nu^{(s)}
	            = 
	            \alpha_{I^{\otimes n}}
	            +
	            \| \alpha \|_{\ell^1}
	            \cdot
	            \sgn(\alpha_P)
	            \cdot
	            \mu(P,\supp(P))
	            \]
        \EndFor
		\Return $\nu = \frac 1 S \sum_s \nu^{(s)}$.
	\end{algorithmic}
\end{algorithm}

For completeness, we record calculations for the expectation and variance of this estimator. 
Consider a single shot giving $\nu$.
Let $\E_P$ denote the expected value over the distribution $\gamma(P)$ and let $\E_{\mu(P)}$ denote the expected value over the measurement outcomes for a fixed Pauli operator $P$. Without loss of generality, we may assume $\alpha_{I^{\otimes n}}=0$. Now $\E_{\mu(P)} \mu(P,\supp(P)) = \tr(\rho P)$ whence
\begin{equation}
    \E(\nu)
    = 
    \E_P \E_{\mu(P)} \nu
    =
    \E_P \| \alpha \|_{\ell^1} \sgn(\alpha_P) \tr(\rho P)
    =
    \sum_P \alpha_P \tr(\rho P)
    =
    \tr(\rho H).
\end{equation}
The variance (for a single sample) can also be calculated:
\begin{align}
    \Var(\nu)
    =
    \sum_{P\neq I^{\otimes n}} ( \gamma_P \cdot \|\alpha \|_{\ell^1}^2 ) - \tr(\rho H_0)^2
    =
    \|\alpha \|_{\ell^1}^2 - \tr(\rho H_0)^2.
\end{align}

\end{subsection}

\begin{subsection}{Largest degree first}

Consider a Hamiltonian decomposed into $K$ collections $\{C^{(k)}\}_{k=1}^{K}$ excluding the identity term: 
$H=\alpha_{I^{\otimes n}}I^{\otimes n} + \sum_{k=1}^K H_k$ 
where $H_k=\sum_{Q\in C^{(k)}}\alpha_Q Q$. 
(Recalling the notation $H_0$ for the traceless part of the Hamiltonian, we note that $H_0 = \sum_k H_k$.)
Suppose that for each collection $C^{(k)}$, the Pauli terms commute qubit-wise: for all $Q,R \in C^{(k)}$ and all qubits $i$, we have $[Q_i,R_i]=0$.
In this case, there exists a Pauli operator $P^{(k)}$ of weight $n$ which commutes qubit-wise with each Pauli in $C^{(k)}$.

Consider also a probability distribution $\kappa$ over the collections $\{C^{(k)}\}_{k=1}^K$. Sampling from this distribution provides Algorithm~\ref{alg:decomposed}.

\begin{algorithm}
    \caption{Energy estimation via decomposition into commuting terms}
    \label{alg:decomposed}
    \begin{algorithmic}
        \For{sample $s \in \{1, \dots, S\}$}
            \State Prepare $\rho$;
            \State Randomly pick collection $C^{(k)}$ from $\kappa$-distribution;
            \For{qubit $i\in\{1,2,\dots, n\}$}
                \State Measure qubit $i$ in $P_i^{(k)}$ basis providing eigenvalue measurement $\mu(P^{(k)}, i) \in \{\pm1\}$;
            \EndFor
            \State Estimate observable expectation
                \[
                \nu^{(s)} 
                = 
                \alpha_{I^{\otimes n}} 
                +
                \frac1{\kappa (C^{(k)})}
                \sum_{Q\in C^{(k)}}
                    \alpha_Q
                    \mu(P^{(k)}, \supp(Q))
                \]
        \EndFor
        \Return $\nu = \frac 1 S \sum_s \nu^{(s)}$.
	\end{algorithmic}
\end{algorithm}
Consider a single sample giving an estimator $\nu$. Similar to the $\ell^1$ algorithm we observe that $\nu$ recovers $\tr(\rho H)$ in expectation. Specifically, for a fixed collection $C^{(k)}$ and hence a fixed full-weight Pauli operator $P^{(k)}$, let $\E_{\mu(P^{(k)})}$ denote the expected value over the measurement outcomes associated with $P^{(k)}$. Now 
$\E_{\mu(P^{(k)})} \mu(P^{(k)}, \supp(Q)) = \tr(\rho Q)$ whenever $Q\in C^{(k)}$
and if we let $\E_{C^{(k)}}$ denote the expected value over the distribution $\kappa(C^{(k)})$ we conclude
\begin{equation}
    \E(\nu) 
    = 
    \E_{C^{(k)}} \E_{\mu(P^{(k)})} \nu
    =
    \E_{C^{(k)}} \frac1{\kappa(C^{(k)})} \sum_{Q\in C^{(k)}} \alpha_Q \tr(\rho Q)
    =
    \tr(\rho H)
\end{equation}
Again, we have assumed without loss of generality that $H$ is traceless.

The variance may be calculated as 
\begin{align}
    \Var(\nu)
    = 
    \left(
    \sum_{k=1}^K
        \frac1{\kappa (C^{(k)})} 
        \sum_{Q,R\in C^{(k)}}
            \alpha_Q \alpha_R
            \prod_{i\in\supp(QR)} \tr(\rho QR)
    \right)
    - \tr( \rho H_0 )^2
\end{align} 
An alternative formula reads \cite[Appendix A]{Kandala2017Hardware-efficientMagnets}
\begin{align}
    \Var(\nu) 
    = 
    \sum_{k=1}^K
        \frac1{\kappa (C^{(k)}) } 
        \sum_{Q,R\in C^{(k)}}
            \alpha_Q \alpha_R
            \left(
                \tr(\rho QR)
                - 
                \tr(\rho Q) \tr(\rho R)
            \right)
\end{align}

Our analysis uses the LDF heuristics in order to obtain such a decomposition.
Various heuristics for building decompositions are investigated in~\cite{verteletskyi2020measurement} for systems up to 36 qubits. The heuristics give numbers of groups that differ by 10\% and they conclude LDF is attractive due to its short runtime.
For the LDF decomposition, we first construct a graph $G=(V,E)$ where:
\begin{itemize}
    \item $v_Q \in V$ for all $Q\neq I^{\otimes n}$ such that $\alpha_Q\neq0$;
    \item $e_{Q,R}\in E$ if $\{Q_i,R_i\}=0$ for some qubit $i$.
\end{itemize}
Second, the vertices of the graph are sorted in decreasing order of their degrees, and the smallest available colour is then progressively assigned to each ordered vertex. Colours correspond to collections in which Pauli operators commute qubit-wise. 
The LDF heuristics guarantee the number of colours of the graph is at most one plus the degree of the graph: $K \le 1 + \Delta(G)$. 
With this decomposition constructed, our analysis is done with the following choice for $\kappa$:
\begin{align}
    \kappa ( C^{(k)} )
    &= 
    \frac
        {\| \alpha |_{C^{(k)}} \|_{\ell^1}}
        {\|\alpha \|_{\ell^1}}
    &
    \| \alpha |_{C^{(k)}} \|_{\ell^1}
    &=
    \sum_{Q \in C^{(k)}} | \alpha_Q |
\end{align}
Qiskit \cite{aleksandrowicz2019qiskit} provides an implementation of the decomposition procedure.\footnote{\url{https://qiskit.org/documentation/_modules/qiskit/aqua/operators/legacy/pauli_graph.html##PauliGraph} Last accessed on June 6, 2020}

\end{subsection}

\end{section}

\end{document}

%% file: table_variances.tex
\begin{table}[!ht]
    
    \caption{\label{tab:algorithms}Variance for various estimators considered in this work. LBCS is optimised according to Eq.~\eqref{eq:cost}, while the diagonal cost function is defined in Eq.~\eqref{eq:cost_diag}.}
    
    \begin{tabular}{| c |  c|  l| 
                           |P{2cm} | }
        
        \hline
        Molecule & Qubits &
        Estimator & {Variance} \\
        \hline
        
        \hline
        \multirow{5}{*}{H$_2$} & \multirow{5}{*}{4} &
        
        $\ell^1$ sampling 
        & \rightvar{ 2.49\hphantom{0} }
        \\&&
        LDF grouping 
        & \rightvar{ 0.402 }
        \\&&
        classical shadows  
        & \rightvar{ 1.97\hphantom{0} }
        \\&&
        LBCS 
        & \rightvar{ 1.86\hphantom{0} }
        \\&&
        LBCS (diagonal cost function) 
        & \rightvar{ 1.86\hphantom{0} }
        \\
        
        \hline
        \multirow{5}{*}{H$_2$} & \multirow{5}{*}{8} &
        
        $\ell^1$ sampling 
        & \rightvar{ 120\hphantom{.000} }
        \\&&
        LDF grouping 
        & \rightvar{ 22.3\hphantom{00} }
        \\&&
        classical shadows  
        & \rightvar{ 51.4\hphantom{00} }
        \\&&
        LBCS 
        & \rightvar{ 17.5\hphantom{00} }
        \\&&
        LBCS (diagonal cost function) 
        & \rightvar{ 17.7\hphantom{00} }
        \\
        
        \hline
        \multirow{5}{*}{LiH} & \multirow{5}{*}{12} &
        
        $\ell^1$ sampling 
        & \rightvar{ 138\hphantom{.000} }
        \\&&
        LDF grouping 
        & \rightvar{ 54.2\hphantom{00} }
        \\&&
        classical shadows  
        & \rightvar{ 266\hphantom{.000} }
        \\&&
        LBCS 
        & \rightvar{ 14.8\hphantom{00} }
        \\&&
        LBCS (diagonal cost function) 
        & \rightvar{ 14.8\hphantom{00} }
        \\
        
        \hline
        \multirow{5}{*}{BeH$_2$} & \multirow{5}{*}{14} &
        
        $\ell^1$ sampling 
        & \rightvar{ 418\hphantom{.000} }
        \\&&
        LDF grouping 
        & \rightvar{ 135\hphantom{.000} }
        \\&&
        classical shadows  
        & \rightvar{ 1670\hphantom{.000} }
        \\&&
        LBCS 
        & \rightvar{ 67.6\hphantom{00} }
        \\&&
        LBCS (diagonal cost function) 
        & \rightvar{ 67.6\hphantom{00} }
        \\
        
        \hline
        \multirow{5}{*}{H$_2$O} & \multirow{5}{*}{14} &
        
        $\ell^1$ sampling 
        & \rightvar{ 4360\hphantom{.000} }
        \\&&
        LDF grouping 
        & \rightvar{ 1040\hphantom{.000} }
        \\&&
        classical shadows  
        & \rightvar{ 2840\hphantom{.000} }
        \\&&
        LBCS 
        & \rightvar{ 257\hphantom{.000} }
        \\&&
        LBCS (diagonal cost function) 
        & \rightvar{ 257\hphantom{.000} }
        \\
        
        \hline
        \multirow{5}{*}{NH$_3$} & \multirow{5}{*}{16} &
        
        $\ell^1$ sampling 
        & \rightvar{ 3930\hphantom{.000} }
        \\&&
        LDF grouping 
        & \rightvar{ 891\hphantom{.000} }
        \\&&
        classical shadows  
        & \rightvar{ 14400\hphantom{.000} }
        \\&&
        LBCS 
        & \rightvar{ 353\hphantom{.000} }
        \\&&
        LBCS (diagonal cost function) 
        & \rightvar{ 353\hphantom{.000} }
        \\
        
        \hline
    \end{tabular}
    
\end{table}

%% file: figure_distributions.tex

\begin{figure}\label{fig:distributions}
\centering
\begin{tikzpicture}
\begin{axis}[
    xbar stacked,
    legend style={  legend columns=3,
                    at={(0.83,-0.055)},
                    draw=none},
    cycle list={    {fill=gray!80,draw=black!50},
                    {fill=gray!50,draw=gray!50},
                    {fill=gray!20,draw=gray!50}},
    ytick=data,
    axis y line*=left,
    axis x line*=bottom,
    tick label style={font=\footnotesize},
    legend style={font=\footnotesize},
    label style={font=\footnotesize},
    xtick={0,0.2,0.4,0.6,0.8,1.0},
    width=0.95\textwidth,
    bar width=6mm,
    xlabel={Probabilities},
    ylabel style={rotate=-90},
    ylabel={Qubit},
    yticklabels={1,2,3,4,5,6,7},
    xmin=0,
    xmax=1,
    area legend,
    y=1cm,
    enlarge y limits={abs=0.6},
]
\addplot coordinates
{(0.4325758758615219, 7)
 (0.24098406063295502, 6)
 (0.13870757752416293, 5)
 (0.13958006241420792, 4)
 (0.16985394780599872, 3)
 (0.267529848293483, 2)
 (0.1164780079544088, 1)};

\addplot coordinates
{(0.4325758758615219, 7)
 (0.24098406063295502, 6)
 (0.13870757752416293, 5)
 (0.13958006241420792, 4)
 (0.16985394780599872, 3)
 (0.267529848293483, 2)
 (0.1164780079544088, 1)};

\addplot coordinates
{(0.13484824827695654, 7)
 (0.5180318787340902, 6)
 (0.7225848449516744, 5)
 (0.7208398751715843, 4)
 (0.6602921043880026, 3)
 (0.4649403034130343, 2)
 (0.7670439840911827, 1)};

\legend{$X$, $Y$, $Z$}
\end{axis}  
\end{tikzpicture}

\caption{Probability distributions over the first 7 of 14 qubits for H$_2$O Hamiltonian using the Jordan-Wigner encoding. The probability distributions have been optimised according to Eq.~\eqref{eq:cost_diag}.}
\end{figure}
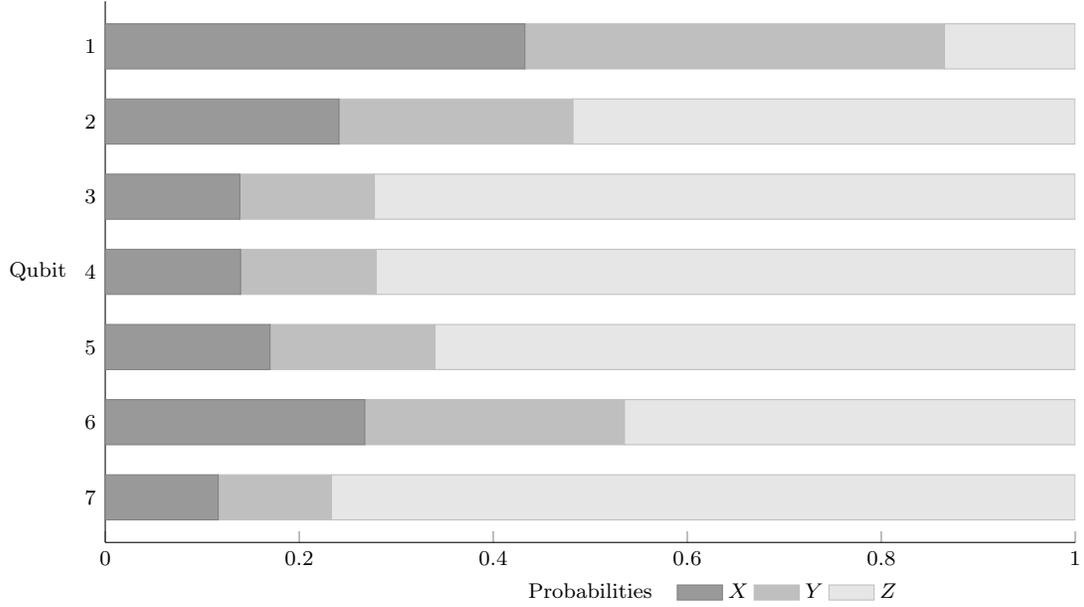

%% file: table_encodings.tex
\begin{table}[!ht]
        
    \caption{\label{tab:encodings}Variance for different estimators computed on the ground states of the molecules indicated. LBCS is optimized with the cost function defined in Eq.~\eqref{eq:cost_diag}.}    
        
    \begin{tabular}{|c | l ||  P{3cm} | P{3cm} | P{3cm} |}
        \hline
        && \multicolumn{3}{c|}{Variance} \\
        \cline{3-5}
        Molecule & Encoding 
        & {LDF grouping}
        & {Classical shadows} 
        & {LBCS} \\
        \hline
        \hline
        \multirow{3}{*}{H$_2$ (4 qubits)} & JW
        & \rightvar{     0.402       }
        & \rightvar{     1.97\hphantom{0}        }
        & \rightvar{     1.86\hphantom{0}        }
        \\
        & Parity
        & \rightvar{     0.193        }
        & \rightvar{     4.00\hphantom{0}        }
        & \rightvar{     0.541        }
        \\
        & BK
        & \rightvar{     0.193        }
        & \rightvar{     10.0\hphantom{00}        }
        & \rightvar{     0.541        }
        \\
        \hline
        \multirow{3}{*}{H$_2$ (8 qubits)} & JW
        & \rightvar{     22.3\hphantom{00}       }
        & \rightvar{     51.4\hphantom{00}        }
        & \rightvar{     17.7\hphantom{00}        }
        \\
        & Parity
        & \rightvar{     38.0\hphantom{00}        }
        & \rightvar{     70.8\hphantom{00}        }
        & \rightvar{     18.9\hphantom{00}        }
        \\
        & BK
        & \rightvar{     38.4\hphantom{00}        }
        & \rightvar{     169\hphantom{.000}        }
        & \rightvar{     19.5\hphantom{00}        }
        \\
        \hline
        \multirow{3}{*}{LiH} & JW
        & \rightvar{     54.2\hphantom{00}       }
        & \rightvar{     266\hphantom{.000}        }
        & \rightvar{     14.8\hphantom{00}        }
        \\
        & Parity
        & \rightvar{     85.8\hphantom{00}        }
        & \rightvar{     760\hphantom{.000}        }
        & \rightvar{     26.5\hphantom{00}        }
        \\
        & BK
        & \rightvar{     75.5\hphantom{00}        }
        & \rightvar{     163\hphantom{.000}        }
        & \rightvar{     68.0\hphantom{00}        }
        \\
        \hline
        \multirow{3}{*}{BeH$_2$} & JW
        & \rightvar{     135\hphantom{.000}       }
        & \rightvar{     1670\hphantom{.000}        }
        & \rightvar{     67.6\hphantom{00}        }
        \\
        & Parity
        & \rightvar{     239\hphantom{.000}        }
        & \rightvar{     3160\hphantom{.000}        }
        & \rightvar{     130\hphantom{.000}        }
        \\
        & BK
        & \rightvar{     197\hphantom{.000}        }
        & \rightvar{     947\hphantom{.000}        }
        & \rightvar{     238\hphantom{.000}        }
        \\
        \hline
        \multirow{3}{*}{H$_2$O} & JW
        & \rightvar{     1040\hphantom{.000}       }
        & \rightvar{     2840\hphantom{.000}        }
        & \rightvar{     258\hphantom{.000}        }
        \\
        & Parity
        & \rightvar{     2670\hphantom{.000}        }
        & \rightvar{     6380\hphantom{.000}        }
        & \rightvar{     429\hphantom{.000}        }
        \\
        & BK
        & \rightvar{     2090\hphantom{.000}        }
        & \rightvar{     10600\hphantom{.000}        }
        & \rightvar{     1360\hphantom{.000}        }
        \\
        \hline
    \end{tabular}
\end{table}

%% file: lbcs.bbl
\begin{thebibliography}{10}

\bibitem{peruzzo2014variational}
Alberto Peruzzo, Jarrod McClean, Peter Shadbolt, Man-Hong Yung, Xiao-Qi Zhou,
  Peter~J Love, Al{\'a}n Aspuru-Guzik, and Jeremy~L O’brien.
\newblock A variational eigenvalue solver on a photonic quantum processor.
\newblock {\em Nature communications}, 5:4213, 2014.

\bibitem{o2016scalable}
P.~J.~J. O'Malley, R.~Babbush, I.~D. Kivlichan, J.~Romero, J.~R. McClean,
  R.~Barends, J.~Kelly, P.~Roushan, A.~Tranter, N.~Ding, B.~Campbell, Y.~Chen,
  Z.~Chen, B.~Chiaro, A.~Dunsworth, A.~G. Fowler, E.~Jeffrey, E.~Lucero,
  A.~Megrant, J.~Y. Mutus, M.~Neeley, C.~Neill, C.~Quintana, D.~Sank,
  A.~Vainsencher, J.~Wenner, T.~C. White, P.~V. Coveney, P.~J. Love, H.~Neven,
  A.~Aspuru-Guzik, and J.~M. Martinis.
\newblock Scalable quantum simulation of molecular energies.
\newblock {\em Phys. Rev. X}, 6:031007, Jul 2016.

\bibitem{Kandala2017Hardware-efficientMagnets}
Abhinav Kandala, Antonio Mezzacapo, Kristan Temme, Maika Takita, Markus Brink,
  Jerry~M. Chow, and Jay~M. Gambetta.
\newblock {Hardware-efficient variational quantum eigensolver for small
  molecules and quantum magnets}.
\newblock {\em Nature}, 2017.

\bibitem{hempel2018quantum}
Cornelius Hempel, Christine Maier, Jonathan Romero, Jarrod McClean, Thomas
  Monz, Heng Shen, Petar Jurcevic, Ben~P Lanyon, Peter Love, Ryan Babbush,
  et~al.
\newblock Quantum chemistry calculations on a trapped-ion quantum simulator.
\newblock {\em Physical Review X}, 8(3):031022, 2018.

\bibitem{jena2019pauli}
Andrew Jena, Scott Genin, and Michele Mosca.
\newblock Pauli partitioning with respect to gate sets.
\newblock {\em arXiv preprint arXiv:1907.07859}, 2019.

\bibitem{yen2020measuring}
Tzu-Ching Yen, Vladyslav Verteletskyi, and Artur~F Izmaylov.
\newblock Measuring all compatible operators in one series of single-qubit
  measurements using unitary transformations.
\newblock {\em Journal of Chemical Theory and Computation}, 16(4):2400--2409,
  2020.

\bibitem{huggins2019efficient}
William~J Huggins, Jarrod McClean, Nicholas Rubin, Zhang Jiang, Nathan Wiebe,
  K~Birgitta Whaley, and Ryan Babbush.
\newblock Efficient and noise resilient measurements for quantum chemistry on
  near-term quantum computers.
\newblock {\em arXiv preprint arXiv:1907.13117}, 2019.

\bibitem{gokhale2019minimizing}
Pranav Gokhale, Olivia Angiuli, Yongshan Ding, Kaiwen Gui, Teague Tomesh,
  Martin Suchara, Margaret Martonosi, and Frederic~T Chong.
\newblock Minimizing state preparations in variational quantum eigensolver by
  partitioning into commuting families.
\newblock {\em arXiv preprint arXiv:1907.13623}, 2019.

\bibitem{zhao2019measurement}
Andrew Zhao, Andrew Tranter, William~M Kirby, Shu~Fay Ung, Akimasa Miyake, and
  Peter Love.
\newblock Measurement reduction in variational quantum algorithms.
\newblock {\em arXiv preprint arXiv:1908.08067}, 2019.

\bibitem{ryabinkin2020iterative}
Ilya~G Ryabinkin, Robert~A Lang, Scott~N Genin, and Artur~F Izmaylov.
\newblock Iterative qubit coupled cluster approach with efficient screening of
  generators.
\newblock {\em Journal of Chemical Theory and Computation}, 16(2):1055--1063,
  2020.

\bibitem{crawford2019efficient}
Ophelia Crawford, Barnaby van Straaten, Daochen Wang, Thomas Parks, Earl
  Campbell, and Stephen Brierley.
\newblock Efficient quantum measurement of pauli operators in the presence of
  finite sampling error.
\newblock {\em arXiv preprint arXiv:1908.06942}, 2019.

\bibitem{hamamura2019efficient}
Ikko Hamamura and Takashi Imamichi.
\newblock Efficient evaluation of quantum observables using entangled
  measurements.
\newblock {\em npj Quantum Information}, 6(1):56, Jun 2020.

\bibitem{preskill2018quantum}
John Preskill.
\newblock Quantum computing in the {NISQ} era and beyond.
\newblock {\em Quantum}, 2:79, 2018.

\bibitem{verteletskyi2020measurement}
Vladyslav Verteletskyi, Tzu-Ching Yen, and Artur~F Izmaylov.
\newblock Measurement optimization in the variational quantum eigensolver using
  a minimum clique cover.
\newblock {\em The Journal of Chemical Physics}, 152(12):124114, 2020.

\bibitem{torlai2019precise}
Giacomo Torlai, Guglielmo Mazzola, Giuseppe Carleo, and Antonio Mezzacapo.
\newblock Precise measurement of quantum observables with neural-network
  estimators.
\newblock {\em Phys. Rev. Research}, 2:022060, Jun 2020.

\bibitem{choo2020fermionic}
Kenny Choo, Antonio Mezzacapo, and Giuseppe Carleo.
\newblock Fermionic neural-network states for ab-initio electronic structure.
\newblock {\em Nature communications}, 11(1):1--7, 2020.

\bibitem{Huang2020PredictingMeasurements}
Hsin-Yuan Huang, Richard Kueng, and John Preskill.
\newblock {Predicting many properties of a quantum system from very few
  measurements}.
\newblock {\em Nature Physics}, 2020.

\bibitem{bonet2019nearly}
Xavier Bonet-Monroig, Ryan Babbush, and Thomas~E O'Brien.
\newblock Nearly optimal measurement scheduling for partial tomography of
  quantum states.
\newblock {\em arXiv preprint arXiv:1908.05628}, 2019.

\bibitem{cotler2020quantum}
Jordan Cotler and Frank Wilczek.
\newblock Quantum overlapping tomography.
\newblock {\em Physical Review Letters}, 124(10):100401, 2020.

\bibitem{dallaire2019low}
Pierre-Luc Dallaire-Demers, Jonathan Romero, Libor Veis, Sukin Sim, and
  Al{\'a}n Aspuru-Guzik.
\newblock Low-depth circuit ansatz for preparing correlated fermionic states on
  a quantum computer.
\newblock {\em Quantum Science and Technology}, 4(4):045005, 2019.

\bibitem{bravyi2019approximation}
Sergey Bravyi, David Gosset, Robert K{\"o}nig, and Kristan Temme.
\newblock Approximation algorithms for quantum many-body problems.
\newblock {\em Journal of Mathematical Physics}, 60(3):032203, 2019.

\bibitem{Boyd2007}
Stephen Boyd, Seung-Jean Kim, Lieven Vandenberghe, and Arash Hassibi.
\newblock A tutorial on geometric programming.
\newblock {\em Optimization and Engineering}, 8(1):67, Apr 2007.

\bibitem{bravyi2017tapering}
Sergey Bravyi, Jay~M Gambetta, Antonio Mezzacapo, and Kristan Temme.
\newblock Tapering off qubits to simulate fermionic hamiltonians.
\newblock {\em arXiv preprint arXiv:1701.08213}, 2017.

\bibitem{Wecker2015ProgressAlgorithms}
Dave Wecker, Matthew~B. Hastings, and Matthias Troyer.
\newblock {Progress towards practical quantum variational algorithms}.
\newblock {\em Physical Review A - Atomic, Molecular, and Optical Physics},
  2015.

\bibitem{arrasmith2020operator}
Andrew Arrasmith, Lukasz Cincio, Rolando~D Somma, and Patrick~J Coles.
\newblock Operator sampling for shot-frugal optimization in variational
  algorithms.
\newblock {\em arXiv preprint arXiv:2004.06252}, 2020.

\bibitem{WelshPowell67}
D.~J.~A. Welsh and M.~B. Powell.
\newblock {An upper bound for the chromatic number of a graph and its
  application to timetabling problems}.
\newblock {\em The Computer Journal}, 10(1):85--86, 01 1967.

\bibitem{aleksandrowicz2019qiskit}
Gadi Aleksandrowicz, Thomas Alexander, Panagiotis Barkoutsos, Luciano Bello,
  Yael Ben-Haim, D~Bucher, FJ~Cabrera-Hern{\'a}ndez, J~Carballo-Franquis,
  A~Chen, CF~Chen, et~al.
\newblock Qiskit: An open-source framework for quantum computing.
\newblock 16, 2019.

\end{thebibliography}
